\documentclass[journal,12pt,a4paper,onecolumn,draftcls]{IEEEtran}
\usepackage[applemac]{inputenc}
\usepackage{lscape}
\usepackage{amsfonts,syntonly}
\usepackage{verbatim,cite,amsmath,amssymb} 
\usepackage{array}
\usepackage{color}
\usepackage{flushend}
\usepackage[pdftex]{graphicx}
\usepackage{stfloats}
\usepackage{amsmath}
 \usepackage{cite}
 
\ifCLASSOPTIONcompsoc
\usepackage[tight,normalsize,sf,SF]{subfigure}
\else
\usepackage[tight,footnotesize]{subfigure}
\fi

\newcommand{\mat}[1]{{\uppercase{\mathbf{#1}}}}
\newcommand{\mats}[1]{{\uppercase{\boldsymbol{#1}}}}
\newcommand{\vects}[1]{{\lowercase{\boldsymbol{#1}}}}
\newcommand{\vect}[1]{{\lowercase{\mathbf{#1}}}}
\newcommand{\Jensen}[1]{\mats{\mathcal{#1}}}

\DeclareMathOperator{\E}{\mathbb{E}}
\DeclareMathOperator{\Gal}{Gal}

\renewcommand{\P}{\text{P}}
\DeclareMathOperator{\out}{out}
\DeclareMathOperator{\Prob}{\mathbb{P}}

\DeclareMathOperator{\PEP}{PEP}

\DeclareMathOperator{\F}{F}
\DeclareMathOperator{\Tr}{Tr}

\DeclareMathOperator{\diag}{diag}
\DeclareMathOperator{\rank}{rank}

\DeclareMathOperator{\id}{\mathbf{I}}

\DeclareMathOperator{\SNR}{SNR}
\renewcommand{\vec}{\text{vec}}

\newtheorem{theorem}{Theorem}
\newtheorem{definition}{Definition}
\newtheorem{lemma}{Lemma}

\begin{document}

\title{How to Achieve the Optimal DMT of Selective Fading MIMO Channels?}
\author{
\IEEEauthorblockN{ 
Lina Mroueh\authorrefmark{2}
and Jean-Claude Belfiore\authorrefmark{1} \\
}
\thanks{Part of this work was performed while the first author was
  with the Communication Theory Goup in ETH Z\"urich.  Lina Mroueh is
  now at the department of signal, image and telecommunications in the
  Institut Sup\'erieur d'Electronique de Paris ISEP,
  France  and
  Jean-Claude Belfiore is with the department of communication and
  electronics in Telecom ParisTech, France.  This paper was presented in part at the IEEE Information
  Theory Workshop (ITW) in Sept. 2010, Dublin, Ireland. }
\IEEEauthorblockA{ %
\IEEEauthorrefmark{2}
ISEP Paris, 75006 Paris, France \\ 
\IEEEauthorrefmark{1}
T\'el\'ecom ParisTech, 75013 Paris, France \\ 
 lina.mroueh@isep.fr and belfiore@enst.fr \\
 }
}

\IEEEcompsoctitleabstractindextext{%
\begin{abstract}
\boldmath
In this paper,  we consider a particular class of selective fading
channel corresponding to a channel that is selective either in
time or in frequency. For this class of channel, we propose a
systematic way to achieve the optimal DMT derived in Coronel and
B\"olcskei, \textit{IEEE ISIT}, 2007 by extending the
non-vanishing determinant (NVD) criterion to the selective channel
case. A new code construction based on split NVD parallel codes is
then proposed to satisfy the NVD parallel criterion.
This result is of significant interest not only in its own right, but also because it settles a long-standing debate in the literature related to the optimal DMT of selective fading channels.
\end{abstract}


}
\maketitle

\IEEEdisplaynotcompsoctitleabstractindextext

\IEEEpeerreviewmaketitle

\section{Introduction and motivations}
\IEEEPARstart{I}n this paper, we consider the selective fading MIMO
channel where a transmitter having $n_t$ antennas wants to communicate
with a receiver having $n_r$ antennas. We assume that the
communication occurs on a channel that exhibits memory either in time
or in frequency. Our objective here is to construct reliable coding
scheme for the high data rate communication in the high SNR regime when the
channel is not known at the transmitter side.  The
performance criteria to evaluate the coding scheme that will be used
in the following is the well-known diversity multiplexing
tradeoff (DMT).  
\par
The diversity multiplexing tradeoff (DMT) proposed by Zheng and Tse in \cite{ZhengTse-2003} is a powerful approach to characterizing the dual benefits in terms of diversity and spatial multiplexing in the high SNR regime. In order to achieve the optimal diversity multiplexing tradeoff for the flat fading MIMO channel, Belfiore \textit{et al.} introduced the non-vanishing determinant criterion in \cite{RekayaBelfiore-2005}. Later, Elia \textit{et al.} \cite{Elia} proved that this criterion is a sufficient condition to achieve the optimal DMT using a full rate code. 
\par
While most of the above results address the case of flat fading
channels, the general channel model of time-frequency selective
channels has been considered by Coronel and B\"olcskei in~\cite{Cor-FS-Jour},~\cite{Coronel-FS} where the optimal DMT is derived. Moreover, a DMT optimal coding scheme based on a joint precoder and parallel codes construction, is proposed. As the block fading channel is a special case of the time-frequency selective channel, it is expected that the DMT expression in\cite{Coronel-FS} matches with the corresponding result in \cite{ZhengTse-2003}. This is, however, not the case\footnote{The optimal DMT expression in \cite{Coronel-FS} is larger than the one in \cite{ZhengTse-2003}. } and has given rise to lots of debate in the literature~\textit{e.g.}~\cite{francis1}. A rigorous interpretation of this incoherence in results remains an open problem. The present paper settles the issue and shows that the DMT derived in \cite{Coronel-FS} is, indeed, achievable. 
\par
\textit{Contributions: }We consider a particular class of the general channel model considered in~\cite{Cor-FS-Jour},\cite{Coronel-FS} where the channel is selective either in time or in frequency. For this class of channels, we propose a systematic way to achieve the optimal DMT by extending the non-vanishing determinant criterion to the selective channel case. A new code construction based on split NVD parallel codes is then proposed to satisfy the NVD parallel criterion. Moreover, for the block fading channel, we provide an extension of the geometrical interpretation to show the achievability of the optimal DMT. This result is of significant interest not only in its own right, but also as it shows that the optimal DMT  in \cite{Coronel-FS} is achievable for all the classes of fading channels including the block fading channel.  
\par
\textit{Outline of the paper: } The rest of the paper is organized as
follows. In Section \ref{sec:model}, we define the selective fading
channel model. We review in Section \ref{sec:prel-bac} some basic preliminaries
and background materials that are essential to the development of this
paper. Then, we derive in Section \ref{sec:pout} the limiting outage bound on
the achievable DMT. We
derive in Section \ref{sec:NVD-codes} the code design criterion required
to achieve this optimal DMT for this class of selective channels and propose a
new family of split NVD parallel codes to satisfy this code design
criterion.  Finally, Section \ref{sec:conclusion} concludes the paper.
\par
\textit{Notation: } The notation used in this paper is as
follows. Boldface lower case letters $\vect{v}$ denote vectors,
boldface capital letters $\mat{M}$ denote matrices. $\mat{M}^\dag$
denotes conjugate transposition. $.^{[T]}$
denote the transposition operator.  $\|\mat{H}\|^2_{\text{F}} =
\Tr\{\mat{H}\mat{H}^\dag\}$ is the Frobenius norm of a
matrix. $\Tr\{\mat{A}\}$ refers to the trace of  matrix
$\mat{A}$. $\id_N$ stands for the $N\times N$ identity
matrix. $\diag\{\mat{A}_n\}_{n=0}^{N-1}$ denotes the block diagonal
matrix containing $\mat{A}_n$ on its diagonal. $\vec{\mat{A}} =
[\vect{a}_1^{[T]} \ldots \vect{a}_N^{[T]}]^{[T]}$, and $\vect{a}_i$ is
a column vector of matrix $\mat{A}$. The non zero eigenvalues of $\mat{A}$ ordered in ascending order are denoted by $\lambda_i (\mat{A})$. $\mathcal{CN}$ represents the complex Gaussian random variable. $\mathbb{E}_{X}$ is the mathematical expectation w.r.t. to the random variable $X$. Equality in distribution between two random variables $X$ and $Y$ is represented by $X \sim Y$. Exponential equality is denoted by $f(x) \doteq x^b$, \textit{i.e.} $\lim_{x\rightarrow \infty} \frac{\log f(x)}{\log x} = b$, and $\dot{\geq}, \dot{\leq}$ denote the exponential inequality. $|\mathcal{A}|$ denotes the cardinality of a set $\mathcal{A}$. Finally, $\mat{A}\otimes\mat{B}$ denotes the Kronecker product of the matrices $\mat{A}$ and~$\mat{B}$.

\section{Channel and signal model}
\label{sec:model}
We consider the general case of selective fading
channel which includes the case of time and frequency selective
channel. In order to deal with such type of channels, techniques that
decompose these channel into parallel sub-channels are generally used
in literature \cite{Durisi2008}. 
The input-output relation for the class of channels considered in this
paper is therefore given by
\begin{equation}
\mat{Y}_n^{[n_r \times T]} = \sqrt{\frac{\SNR}{n_t}}\mat{H}_n^{[n_r \times n_t]}\mat{x}_n^{[n_t \times T]} + \mat{Z}_n^{[n_r \times T]},
\label{eq:model}
\end{equation}
where $n= 0,1,\ldots,N-1$ represents the sub-channel $n$, the sub-channel $\mat{H}_n^{[n_r \times n_t]}$ is a $n_t\times n_r$ MIMO channel that remains constant during all the duration of the transmission $T$, $\mat{X}_n$ represents the transmitted signal, and $\mat{Z}_n$ denotes the additive i.i.d. $\mathcal{CN}(0,\id)$ noise. The channels $\mat{H}_n$ are correlated across the sub-channels $n = 0 \ldots N-1$ according to, 
\begin{equation}
 \mat{H} = [\mat{H}_0 \; \ldots \; \mat{H}_{N-1}] = \mat{H}_{w} (\mat{R}^{1/2}_{\mathbb{H}} \otimes \id_{n_t}),
\label{eq:mat-corr}
\end{equation}
where $\mat{R}_{\mathbb{H}}$ is the $N\times N$ correlation between
the scalar sub-channels characterized by its rank equal to $\rho \leq N$, $\mat{H}_{w}$ is an  $n_r \times Nn_t$  
matrix with i.i.d. $\mathcal{CN}(0,1)$ entries.
The transmitted signal satisfies the following power constraint, 
\begin{equation}
\sum_{i=0}^{N-1}\E\big[\|\mat{X}_{i}\|^2 _{\F}\big] \leq TN.
\label{eq:pow-const}
\end{equation}
Throughout this paper, we set $m=\min(n_t,n_r)$ and $M = \max(n_t,n_r)$.
\par
The input-output relation considered in (\ref{eq:model}) models the case when the channel is selective either in time or in frequency. 
For the frequency selective channel, the MIMO OFDM system decomposes
the channel into $N$ parallel subcarrier, where $N$ represents the
total number of subcarriers and $n$ stands for the frequency. The
sub-channel remains constant over each subcarrier and the correlation
matrix $\mat{R}_{\mathbb{H}}$ is a circulant matrix having a rank
equal to $L$ which is nothing but the number of channel taps or the
\emph{memory} of the selective channel.  
\par
For the time selective case (or the block fading channel), the channel
remains constant during a block~$n$ of $T$ time slots and changes in a
statistically independent manner across blocks. For this case, $N$
represents the total number of blocks and~$\mat{R}_{\mathbb{H}} =
\id_N$ with full rank~$N$.

\section{Preliminaries and background}
\label{sec:prel-bac}
In this section, we start by 
recalling some basic preliminaries on the optimal diversity multiplexing
tradeoff (DMT) of the code in Subsection
\ref{sec:prel} and on the limiting outage bound of the selective
fading channel in Subsection \ref{eq:DMT-sel}. Then, we briefly review
prior results from literature that motivate our contribution. 

\subsection{Diversity multiplexing tradeoff (DMT)}
\label{sec:prel}
Let $\mathcal{X}_p(\SNR)$ be a family of coding schemes operating at a given $\SNR$, and let $R(\SNR)$ denote the rate transmitted \emph{per sub-channel}, such that, 
$$
R(\SNR) = r \log \SNR,
$$
where $r$ is the multiplexing gain \emph{per sub-channel}. 
\par
The diversity multiplexing tradeoff (DMT) of the coding scheme $\mathcal{X}_p(\SNR)$ is defined as the $\SNR$ exponent of the error probability $ P_{e,\mathcal{X}_p}(r,\SNR)$ using maximum likelihood-decoding such that 
\begin{equation*}
d(r) = -\displaystyle \lim_{\SNR \rightarrow \infty} \frac{\log P_{e,\mathcal{X}_p}(r,\SNR)}{\log\SNR}.
\end{equation*}
For a given multiplexing gain $r$, the optimal DMT is the largest DMT supported by any coding scheme, and is is denoted by $d^*(r)$. 

\subsection{DMT outage bound}
\label{eq:DMT-sel}
\par
The outage probability of a selective fading channel when the target rate $R$ scales as $r\log\SNR$ is defined as, 
$$
P_{\out}(r) \triangleq \Prob\Big\{\log\det\Big(\id_N +
\frac{\SNR}{n_t}\mats{\mathcal{H}}\mats{\mathcal{H}}^\dag \Big) < Nr
\log\SNR \Big\},
$$
where $\mats{\mathcal{H}} = \diag\{\mat{H}_{n}\}_{n=0}^{N-1}$ is the block diagonal channel matrix.
\\*
The optimal DMT of the selective fading MIMO channel
has been derived in \cite{Cor-FS-Jour} and \cite{Coronel-FS}.  For
this general case, Coronel and B\"olcskei showed that the
outage probability is bounded as, 
\begin{equation}
 P_{e,\mathcal{X}_p}(r) \geq  P_{\out}(r) \geq P_J(r) \doteq \SNR^{-d_{J}(r)}
\label{eq:ineq}
\end{equation}
where, 
\begin{equation}
d_J(r) = (\rho M-r)(m-r).
\label{eq:DMT-freq}
\end{equation}
and $m = \min(n_t,n_r)$ and $M = \max(n_t,n_r)$.
Note that the first inequality in (\ref{eq:ineq}) is a consequence of the Fano
inequality \cite{ZhengTse-2003}, and the second inequality is a
consequence of the Jensen inequality as shown in \cite{Cor-FS-Jour}.
Moreover, a coding scheme that achieves the bound called "Jensen
bound" in the terminology of \cite{Coronel-FS}  has been proposed
in \cite{Cor-FS-Jour} and \cite{Coronel-FS}.  It follows therefore
from\cite{Coronel-FS} that the optimal DMT is equal to, 
$$
d(r) = d_J(r) = (\rho M-r)(m-r).
$$ 

\subsection{Previous work and motivations}
 The block fading channel is a particular case of the selective fading
 channel model considered in (\ref{eq:model}) with covariance matrix
 $\mat{R}_{\mathbb{H}} = \id_{N}$. The optimal DMT expression is
 therefore $d^*(r) = (NM-r)(m-r)$, which is the DMT expression of the
 general channel model considered
 in~\cite{Cor-FS-Jour},\cite{Coronel-FS} applied to this particular
 channel setting. Obviously, this result does not match with the
 corresponding result in \cite{ZhengTse-2003}, \textit{i.e.}, $d_l(r)
 = N(M-r)(m -r) \leq d^*(r), \; \forall r$. This incoherence in
 results has been subject to lots of debate in literature
 \textit{e.g.} \cite{francis1} and motivates our contribution. The authors of \cite{francis1} base
 their arguments on a non-accurate outage probability derivation
 ($P_{\out,l}(r) \doteq \SNR^{-d_l(r)}$) to claim that the DMT of the
 block fading channel cannot exceed $d_l(r) \leq d^*(r)$. 
In order to settle this issue,  we show in
this paper that the DMT in \cite{ZhengTse-2003} is not a limiting
outage bound as claimed in \cite{francis1}, and that the DMT in \cite{Coronel-FS} is achievable using codes
derived from cyclic division algebra (CDA). 

\section{Outage bound on the DMT of selective fading channel}
\label{sec:pout}
Unlike the flat fading channel, the analytical outage probability for
the selective fading channel cannot be easily derived using the
eigenvalues distribution. For the case of correlated parallel
sub-channels, Coronel and B\"olcskei in \cite{Cor-FS-Jour} generalize the geometrical
interpretation in \cite{ZhengTse-2003} to the selective fading case. For the particular case
of the statistically independent parallel sub-channels which is the block fading
channel, the analytical outage probability should be carefully
performed to take into account the impact of coding across the blocks,
which cannot be easily seen using the block diagonal structure of the
matrix. For this, an equivalent expression of the outage probability is first
derived. Then, we provide here an outage derivation based on the geometrical argument previously used for the flat fading channel in \cite{ZhengTse-2003} and for the selective fading case in \cite{Cor-FS-Jour}.  

\subsection{Outage bound of the block fading channel}
For the block fading channel, the outage probability is, 
\begin{equation*}
P_{\out}(r)  \triangleq  \Prob\Big\{\log\det\Big(\id + \frac{\SNR}{n_t}\mats{\mathcal{H}}\mats{\mathcal{H}}^\dag \Big) < Nr \log\SNR \Big\},
\end{equation*}
where $\mats{\mathcal{H}} = \diag\{\mat{H}_{n}\}_{n=0}^{N-1}$ is the
block diagonal channel matrix. 

\subsubsection{Equivalent outage expression}
In order to generalize the geometrical interpretation in~\cite{Cor-FS-Jour} to the block fading channel, we start first by
finding in Lemma \ref{lemma:Eq-DMT} an equivalent expression of the outage
probability.  
\begin{lemma}
For the block fading channel, the outage probability is equivalent to,
\begin{equation}
\text{P}_{\text{out}}(r)  
= \mathbb{P}\Big\{\log\det\big(\mathbf{I} +
\frac{\text{SNR}}{Nn_t}\mathbf{C}_{\text{H}}^{\;}\mathbf{C}_{\text{H}}^\dag
\big) < Nr \log\text{SNR} \Big\},
\label{eq:out-def}
\end{equation}
where
\begin{equation}
\mathbf{C}_{\text{H}}= \left[ \begin{array}{cccc}
\mathbf{H}_{w,0} & \mathbf{H}_{w,1}  & \ldots &  \mathbf{H}_{w, N-1} \\
 & & \vdots  & \\
\mathbf{H}_{w,1}& \mathbf{H}_{w,2}& \ldots &  \mathbf{H}_{w, 0} 
\end{array}
\right]
\label{eq:Ch}
\end{equation}
and $ \mathbf{H}_{w,i}$, $i = 0 \ldots N-1$ are Gaussian matrices with
i.i.d. entries.
\label{lemma:Eq-DMT}
\end{lemma}

Before going to the rigorous proof, we note here that the main intuition behind
this lemma is the fact that the block fading channel can be considered as a
selective fading channel with a channel memory of $N$ blocks. This 
is so far the case as the covariance matrix is equal to identity,
which is a full rank matrix with rank equal to $N$.

\begin{proof}
To prove this lemma, we consider $\vect{h}_{ij}$ the $N \times 1$ Gaussian vector $\sim
\mathcal{CN}(0,\id_N)$ containing the $N$ independent channel
realisations between transmit antenna $j$ and receive antenna $i$.  It
is well-known that the Gaussian vector $\vect{h}_{ij}$ is identically distributed
as $\mat{F}\vect{h}_{\omega,ij}$ for any unitary matrix $\mat{F}$, \textit{i.e.},
$
\vect{h}_{ij} \sim \mat{F}\vect{h}_{\omega,ij}, \forall i,j.
$

In the following, we specify our result to the case where $F$ is a $N
\times N$ Fast Fourier Transform (FFT) matrix.  This means that each channel realisation is
identically distributed as,
$$
h_{ij}^{[n]} \sim \frac{1}{\sqrt{N}}\sum_{l=0}^{N-1} h_{ij,w}^{[l
]} e^{-j 2\pi \frac{ln}{N}}, \quad n = 0 \ldots N-1.
$$
The block diagonal matrix $\mats{\mathcal{H}}$ is therefore
identically distributed as $\mat{D}_{\text{H}}$, \textit{i.e.}, 
$
\mats{\mathcal{H}}\sim \mat{D}_{\text{H}},
$
where,
\begin{equation}
\mat{D}_{\text{H}}= \frac{1}{\sqrt{N}} \left[ \begin{array}{ccc}
\displaystyle\sum_{l=0}^{N-1}\mat{H}_{w,l} \omega_l^{0}  &  &  \\
&\ddots& \\
 & & \displaystyle \sum_{l=0}^{N-1}\mat{H}_{w,l} \omega_l^{N-1}
\end{array}
\right],
\label{eq:Dh}
\end{equation}
\normalsize
with $\omega_l = e^{-j\frac{2\pi l}{N}}$ and $\mat{H}_{\omega,l} =
(h_{ij,\omega}^{[l]})_{ 1 \leq i  \leq n_r, 1 \leq j \leq n_t}$. 

Consequently, the mutual information is identically distributed~as, 
$$
I(\vect{x},\vect{y}|\mat{H}) \sim \log \det \Big(\id  +
  \frac{\SNR}{Nn_t}\mat{D}_{\text{H}}\mat{D}_{\text{H}}^\dag \Big) = I_D(\SNR).
$$
By using an FFT precoder and an FFT equalizer as in an OFDM system
to transmit over the channel $\mat{D}_{\text{H}}$ in (\ref{eq:Dh}), the matrix
$\mat{D}_{\text{H}}\mat{D}{\text{H}}^{\dag}$ can be made unitarily
equivalent to $\mat{C}_{\text{H}}\mat{C}_{\text{H}}^\dag$, where
$$
\mat{C}_{\text{H}}= \left[ \begin{array}{cccc}
\mat{H}_{w,0} & \mat{H}_{w,1}  & \ldots &  \mat{H}_{w, N-1} \\
\mat{H}_{w,N-1}& \mat{H}_{w,0}& \ldots &  \mat{H}_{w, N-2} \\
 & & \vdots  & \\
\mat{H}_{w,1}& \mat{H}_{w,2}& \ldots &  \mat{H}_{w, 0} 
\end{array}
\right].
$$
Thus, the corresponding mutual information $I_D(\SNR)$ can be written as, 
$$
 I_D(\SNR) = \log \det \left(\id  +
  \frac{\SNR}{Nn_t}\mat{C}_{\text{H}}\mat{C}_{\text{H}}^\dag \right) \sim I(\vect{x},\vect{y}|\mat{H}).
$$
It follows therefore that the outage probability is such that, 
$$
\P_{\text{out}}(r) =  \Prob\Big\{\log\det\Big(\id +   \frac{\SNR}{Nn_t}\mat{C}_{\text{H}}\mat{C}_{\text{H}}^\dag \Big) < Nr \log\SNR \Big\}.
$$
\end{proof}

\subsubsection{Geometrical interpretation}
Following the geometrical interpretation of the flat fading channel in \cite{ZhengTse-2003}, the typical outage event occurs when the channel matrix $\mat{C}_{\text{H}}$ is close to the manifold of all matrices with rank $Nr$ denoted by $\mathcal{R}_{Nr}$, such that, 
$$
\mathcal{R}_{Nr} = \{\mat{C}_{\text{H}}: \rank\{\mat{C}_{\text{H}}\} = Nr\}. 
$$
By following the same reasoning as in \cite{ZhengTse-2003}, this requires that the $d(r)$ components of $\mat{C}_{\text{H}}$ orthogonal to $\mathcal{R}_{Nr}$ to be collapsed, \textit{i.e.}, be on the order of $\SNR^{-1}$. The probability of this event is $P_{\out}(r) \doteq \SNR^{-d(r)}$. The number of these components is given by 
$$
d(r) = NMm - \dim (\mathcal{R}_{Nr}),
$$
where $\dim (\mathcal{R}_{Nr})$ is the sufficient \emph{minimal} number of parameters
required to specify matrix $\mat{C}_{\text{H}}$ with rank $Nr$.

\subsubsection{Dimensionality of $\mathcal{R}_{Nr}$}
We first note that due to the structure of $\mat{C}_{\text{H}}$ in
(\ref{eq:Ch}), the number of parameters required to characterize a
matrix $\mat{C}_{\text{H}}$ in $\mathcal{R}_{Nr}$ is equal to the
number of parameters required to specify an $m \times NM$
matrix ($m = \min(n_t,n_r)$ and $M = \max(n_t,n_r)$) with rank $r$ that contains the $n_t$ first columns if $n_t
\leq n_r$, and the $n_r$ first rows if $n_r \leq n_t$ as shown in
Figures \ref{fig:dim2} and \ref{fig:dim1} .
 
\begin{figure*}[htbp]
\centering
\centerline{\subfigure[Case 1: $n_r \leq n_t$]{\includegraphics[width= 0.45\linewidth]{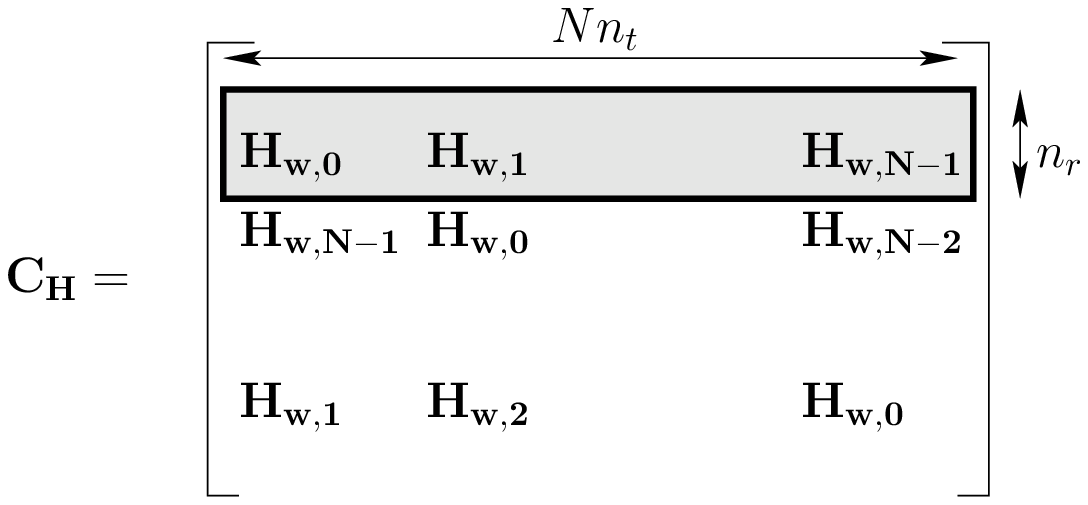}
\label{fig:dim2}}
\hfill
\subfigure[Case 2: $n_t \leq
n_r$]{\includegraphics[width=0.4\linewidth]{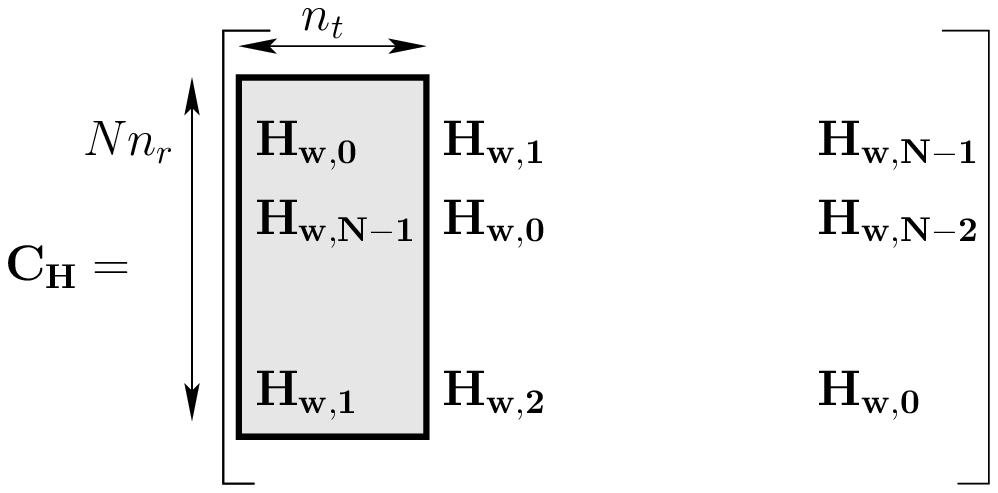}
\label{fig:dim1}}}
\caption{It is sufficient to specify a $m \times
  NM$ matrix with rank $r$, with  $m = \min(n_t,n_r)$ , $M  =
  \max(n_t,n_r)$ to characterize a matrix $\mat{C}_{\mathbb{H}}$ with
  rank $Nr$.}
\label{fig:dim}
\end{figure*}

Characterizing a matrix $\mat{C}_{\text{H}}$ with rank $Nr$ reduces
therefore to the problem of characterizing a
  matrix of dimension $m \times NM$ with rank $r$ that requires only
  $NMr + (m-r)r$, \textit{i.e},
$$
 \dim(\mathcal{R}_{Nr}) = NMr + (m-r)r,
$$
 where $MNr$ is the number of independent parameters
  needed to identify $r$ independents vectors and $(m-r)r$ parameters
  are needed to identify the linear dependent vectors as a function of
  the $r$ independent vectors.  
It can be be easily verified here that the $MNr$ free i.i.d. Gaussian parameters
that identify the~$r$ linear independent 
vectors generate a block circulant matrix with rank $Nr$ with a
probability equal to one.  

It can be deduced that the optimal DMT for the class of block fading channel is,
$$
d_{out}(r) = NMm - \dim(\mathcal{R}_{Nr}) = (NM -r)(m-r).
$$

\subsection{Comments on related work's derivation}
It turns out fron the geometrical interpretation that the outage event
is reduced to the probability that the $m \times NM$ Jensen channel, denoted by $\Jensen{H}_w$ in the rest of the paper, is
in outage, which is the Jensen outage event in the Coronel and
B\"olcskei terminology \cite{Coronel-FS}. 
This means that the outage event is reduced to, 
$$
\mathcal{O}(\SNR) = \{\Jensen{H}_w \in \mathbb{C}^{m\times NM}~\text{is in outage}\}
$$
Note that the straightforward generalization of the flat fading outage
results to the block diagonal matrix in (\ref{eq:out-def}) as in \cite{ZhengTse-2003} and
\cite{francis1} does not take into account the impact of the coding
among the channel blocks in the analytical outage derivation and does
not lead to an accurate outage probability expression. In the
following, we show how this optimal DMT can be achieved using a code
derived from cyclic division algebra (CDA).

\section{DMT achievability: Split NVD parallel codes for selective fading channel}
\label{sec:NVD-codes} 
In this section, we propose a new family of split NVD parallel codes
to achieve the optimal DMT of $(\rho M-r)(m-r)$. We start first by
deriving in Subsection \ref{sec:suf-cond} a sufficient condition on the code to achieve
the optimal DMT for this class of channel. The new family of proposed
codes is based on the previously known NVD parallel codes family which
we will briefly review in Subsections \ref{sec:nvd-codes1} and \ref{sec:nvd-codes2}. Finally, the code construction and the optimality of the split NVD parallel code is addressed in Subection \ref{sec:splitted-code}. 

\subsection{Optimal code design criterion} 
\label{sec:suf-cond}
Unlike the case of time-frequency selective channel in \cite{Cor-FS-Jour}, we show here that when the channel is selective either in time or in frequency, there is no need to construct an additional precoder adapted to the channel statistics in order to achieve the optimal DMT. 
The optimal code design criterion required to achieve the optimal DMT is summarized in the following theorem. 
\begin{theorem}[Sufficient condition for DMT achievability] 
A coding scheme $\mat{X} \in \mathcal{X}_p(\SNR)$ achieves the optimal DMT $(\rho M-r)(m-r)$, if for any two different codewords $\mat{X},\hat{\mat{X}} \in \mathcal{X}_p(\SNR)$, the eigenvalues of the block diagonal matrix $\mat{DD}^\dag$, where
$
\mat{D} = \diag\big\{(\mat{X}_n - \hat{\mat{X}}_n)\big\}_{n=0}^{N-1}
$  
satisfy
\begin{equation}
\min_{\mat{X},\hat{\mat{X}} \in \mathcal{X}_p(\SNR)} \; \prod_{i=1}^{m} \lambda_{i}(\mat{D}\mat{D}^\dag) ~\dot{\geq}~ \frac{1}{2^{R(\SNR)+o(\SNR)}}.
\label{eq:nvd-o}
\end{equation}
\label{theo:suf-cond}
\end{theorem}

\begin{proof}
The proof of this theorem uses the same steps as the proof of [Theorem 1 in \cite{Cor-FS-Jour}] and is detailed in Appendix~\ref{ap:proof-theo2}.

\end{proof}

\subsection{NVD parallel scheme}
\label{sec:nvd-codes1}
Let $\mat{X} = \diag\{\mat{X}_n\}_{n = 0}^{N-1} \in \mathcal{X}_p(\SNR)$ be the block diagonal matrix containing the transmitted codeword $\mat{X}_i$ in (\ref{eq:model}), and constructed such that 
$
\mat{X} = \theta \; \mats{\Xi},
$
where $\theta$ is a scaling factor that depends on the structure of the code, and chosen to ensure the power constraint in (\ref{eq:pow-const}). The block diagonal matrix $\mats{\Xi} = \diag\{\tilde{\mats{\Xi}}_i\}_{i=0}^{N-1}$ is an NVD parallel code denoted by $ \mathcal{C}(\SNR)$, and defined as follows: 
\begin{definition}[NVD parallel scheme]
Let $\mathcal{A}(\SNR)$ be an alphabet\footnote{We assume here without restriction that the signal constellation is a QAM constellation, \textit{i.e}, $\mathcal{A}(\SNR) = \mathcal{A}_{\text{QAM}}(\SNR)$. This can be also extended to the case of HEX constellations.} that is salably dense, such that
\begin{eqnarray*}
\forall s \in \mathcal{A}(\SNR) &\Rightarrow& |s|^2~\dot{\leq}~|\mathcal{A}(\SNR)|.
\end{eqnarray*}
Then, $\mathcal{C}(\SNR)$ is called NVD parallel code if,
\begin{enumerate}
\item Each entry of $\mats{\Xi}$ is a linear combination of symbols carved from $\mathcal{A}(\SNR)$. 
\item The total number of transmitted symbols carved from $\mathcal{A}(\SNR)$ is equal to $TNn_t$. 
\item For any pair of different codewords $\mats{\Xi}$ and $\hat{\mats{\Xi}} \in \mathcal{C}(\SNR)$, the NVD property is satisfied
\begin{equation}
\det\big((\mats{\Xi}-\hat{\mats{\Xi}})(\mat{\Xi}-\hat{\mat{\Xi}})^\dag\big) ~\geq~ \kappa > 0, 
\label{eq:NVD-CDA}
\end{equation}
with $\kappa$ is a constant independent of $\SNR$.
\end{enumerate}
\label{def:nvd-codes}
\end{definition}

\subsection{Cyclic division algebra (CDA) code structure}
\label{sec:nvd-codes2}
We recall here the most relevant concepts of the construction of the codeword  matrix $\mats{\Xi} = \diag\{\tilde{\mats{\Xi}}_i\}_{i=0}^{N-1}$ based on cyclic division algebra. We refer the reader to \cite{Sheng-Par-Chan},\cite{francis2} for more details on the NVD parallel code construction. 
In the following, we consider,
\begin{itemize}
\item[-] The field $\mathbb F$ as a Galois extension of degree $N$ over $\mathbb{Q}(i)$, and that have $\tau$ as generator, such that
$$
\Gal(\mathbb{F}/\mathbb{Q}(i)) = \{\tau_0,\ldots,\tau_{N-1}\}.
$$
\item[-] The field $\mathbb{K}$ is a cyclic extension of degree $n_t$ over $\mathbb F$, and that have $\sigma$ as generator, such that
$$
\Gal(\mathbb K/\mathbb F) = \{\sigma^0,\ldots,\sigma^{n_t-1}\}.
$$
\end{itemize}
The code $\mats{\Xi}$ is constructed by setting $\tilde{\mats{\Xi}}_i = \tau_i(\mats{\tilde{\Xi}})$, \textit{i.e.},
\begin{equation}
\mats{\Xi} = \left[
\begin{array}{cccc}
\mats{\tilde{\Xi}}& & & \\
&\tau_1(\mats{\tilde{\Xi}}) & & \\
& & \cdots & \\
& & & \tau_{N-1}(\mats{\tilde{\Xi}})
\end{array}
\right]
\end{equation}
where $\mats{\tilde{\Xi}}$  belongs to the cyclic division algebra $\mathcal{C} = (\mathbb{K}/\mathbb{F}, \sigma,\gamma)$, and $\gamma \in \mathbb F$ chosen such that $\gamma,\gamma^2,\ldots,\gamma^{n_t-1}$ are not norms of an element of $\mathbb K$.
The matrix $\mats{\tilde{\Xi}}$ is defined such that
$$ 
\mats{\tilde{\Xi}} = \left(
		 \begin{array}{cccc}
		 x_0 & x_1 & \ldots & x_{n_t-1}\\
		 \gamma \sigma(x_{n_t-1}) & \sigma(x_0) & \ldots & \sigma(x_{n_t-2})\\
		 \vdots & ~ & ~ & \vdots \\
		 \gamma \sigma^{n_t-1}(x_1)& \gamma \sigma^{n_t-1}(x_2) & \ldots  & \sigma^{n_t-1}(x_0)		
		 \end{array}
		 \right),
$$
where,
$
x_i = \sum_{j=1}^{Nn_t} s_{i,j} \omega_{j}, \quad s_{i,j} \in \mathcal{A}(\SNR) \quad \text{and}~~\omega_j \in \mathbb{K}.
$
For the NVD parallel code, the determinant is such that,
\begin{align*}
\det \big(\diag\{\tilde{\mats{\Xi}}_i \}_{i=1}^N\big) &= \prod_k \tau_k(\det(\tilde{\mats{\Xi}}))  \notag \\
&= N_{\mathbb F / \mathbb{Q}(i)}(\det(\tilde{\mats{\Xi}}_i ))  \in \mathbb Z[i],
\end{align*}
and which is equal to zero if and only if all $x_i$ are zeros.
It follows that for $\mats{\Xi} \neq \bf{0}$ ,
$$
|\det(\mats{\Xi})|^2~\dot{\geq}~\SNR^0.
$$
 
We finally recall that the NVD parallel codes preserve the mutual
information as, 
$$
\vec\Big(\Big[ \mats{\tilde{\Xi}}^{[T]} \; \ldots
 \; \tau_{N-1}(\mats{\tilde{\Xi}}) ^{[T]}   \Big]^{[T]}\Big) = \mats{\Phi}\;\vect{s}
$$
where $\mats{\Phi}$ is an orthogonal matrix, such that
$\mats{\Phi\Phi}^\dag = \id_{Nn_t} $. 
It follows therefore that the mutual information between the
vectorized input vectors $\tilde{\vect{X}} = \vec(\left[
  \mat{X}_0^{[T]} \ldots \mat{X}_{N-1}^{[T]} \right])$ and the
vectorized output $\tilde{\vect{Y}} = \vec(\left[
  \mat{Y}_0^{[T]} \ldots \mat{Y}_{N-1}^{[T]} \right])$ is,
$$
I(\tilde{\vect{x}},\tilde{\vect{Y}}|\mat{H}) = \log\det\Big(\id_N + \frac{\SNR}{n_t}\mats{\mathcal{H}}\mats{\mathcal{H}}^\dag \Big),
$$ 
where $\mats{\mathcal{H}} = \diag\{\mat{H}_{n}\}_{n=0}^{N-1}$ is the block diagonal channel matrix.

\subsection{Choice of $\theta$ for NVD parallel codes}
Following the same reasoning in \cite{Elia} and \cite{francis1}, the scaling factor $\theta$ that insures the power constraint in (\ref{eq:pow-const}) is such that,
$$
\theta^2 \sum_{i = 0}^{N-1} \E[\|\tilde{\mats{\Xi}}_i\|^2_{\F}] \leq TN.
$$
Due the linearity of this code and to the use of unit transformation, each entry of  $x \in \mats{\Xi}$ is such that, 
\begin{eqnarray*}
\E[|x|^2] &=& \E[|s|^2], \quad s \in \mathcal{A}_{\text{QAM}}(\SNR), \\
                &=& \frac{2(|\mathcal{A}(\SNR)|-1)}{3}.
\end{eqnarray*}
This implies that, 
\begin{eqnarray*}
\sum_{i = 0}^{N-1} \E[\|\tilde{\mats{\Xi}}_i\|^2_{\F}] &=& TN\E[|x|^2] , \\
                          &\doteq& TN|\mathcal{A}(\SNR)|.
\end{eqnarray*}
The scaling factor $\theta$ that ensures the power constraint is therefore, 
\begin{equation}
\theta^2 ~\doteq~ |\mathcal{A}(\SNR)|^{-1}.
\label{eq:theta}
\end{equation}

Using the NVD parallel criterion in (\ref{eq:NVD-CDA}) and the value of $\theta^2$ in~(\ref{eq:theta}), the eigenvalues of the block diagonal matrix $\mat{D}=  \mat{X} - \hat{\mat{X}} = \theta (\mats{\Xi} - \hat{\mats{\Xi}} )$ for any different codewords $\mat{X},\hat{\mat{X}}$, are such that,
$$
\prod_{i=1}^{Nn_t} \lambda_{i}(\mat{D}\mat{D}^\dag) = \frac{|\det(\mats{\Xi}-\hat{\mats{\Xi}})|^2 }{|\mathcal{A}(\SNR)|^{Nn_t}}~\dot{\geq}~  \frac{1}{|\mathcal{A}(\SNR)|^{Nn_t}}. 
$$
Due to the power constraint in (\ref{eq:pow-const}), these eigenvalues necessarily satisfy  $\lambda_{i}(\mat{D}\mat{D}^\dag) ~\dot{\leq}~1$. 
Then, the NVD parallel criterion is equivalent to,
\begin{equation}
\min_{\mat{X},\hat{\mat{X}} \in \mathcal{X}_p(\SNR)} \; \prod_{i=1}^{m} \lambda_{i}(\mat{D}\mat{D}^\dag) ~\dot{\geq}~ \frac{1}{|\mathcal{A}(\SNR)|^{Nn_t}}. \label{eq:gen-nvd}
\end{equation}
It can be easily verified that the NVD parallel criteria of the NVD
parallel code depends critically on the size of the constellation. The
natural question that comes here is: What is the optimal size of
constellation that guarantees to transmit a rate $R(\SNR)$ over each
sub-channel and that meets the sufficient condition of DMT
achievability in (\ref{eq:nvd-o}).

\subsection{Split NVD parallel codes and optimality}
\label{sec:splitted-code}
The NVD parallel codes as put straightforwardly by Lu
in~\cite{francis2} and Yang \textit{et} al. in~\cite{Sheng-Par-Chan}
are sub-optimal, as the DMT achieved by these codes is only $\rho(n_t
- r)(n_r - r) < (\rho M -r)(m-r)$. The main idea of the new split code
construction is to design a coding scheme that guarantees to transmit a rate of
$R(\SNR)$ using a total power of $\SNR$ over each
sub-channel and to satisfy the NVD parallel
criterion in Theorem~\ref{theo:suf-cond}. The two possible ways of
splitting the data over the parallel channels are detailed in
Subsections \ref{sec:BD-NVD} and \ref{sec:split-NVD}. 

\subsubsection{Block diagonal NVD parallel code}
\label{sec:BD-NVD} 
The first way of splitting the data over the parallel channels has been
previously studied in \cite{francis2} and is depicted in Figure
\ref{fig:TS}.  
\begin{figure}[htbp]
\centering
\includegraphics[width = \linewidth]{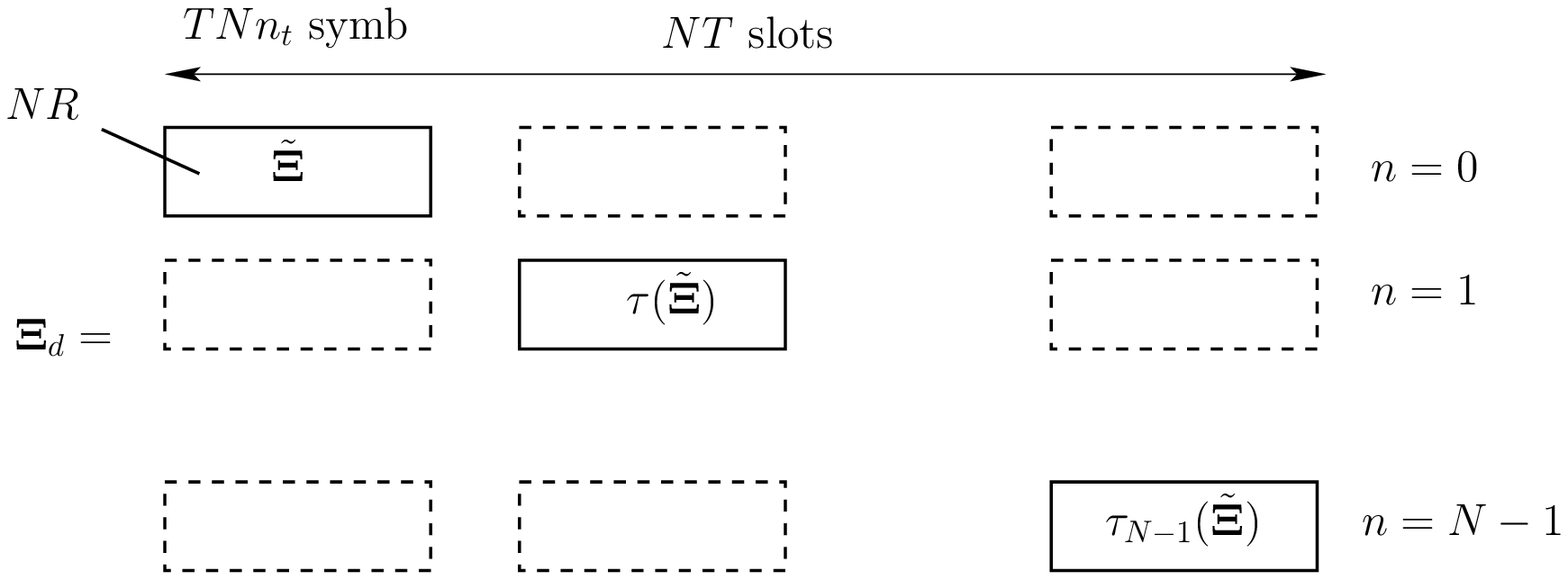}
\caption{Coding across time and frequency: The total rate is transmitted only during $T$ slots. Each entry of $\tau_i(\tilde{\mats{\Xi}})$ is a linear combination of symbols carved from $\mathcal{A}_d(\SNR)$ where $|\mathcal{A}_d(\SNR)| = \SNR^{\frac{r}{n_t}}$. In this case, $\mat{X}_{e,d} = \theta_d \mats{\Xi}_d$.}
\label{fig:TS}
\end{figure}

In this case, the total rate $NR$ is transmitted during only $T$
 slots over each sub-channel.
It can be easily verified that for this scheme the outage event is
such that, 
$$
\mathcal{O}_1 (r,\SNR) = \left\{
  I_1(\tilde{\vect{x}},\tilde{\vect{Y}}|\mat{H})  < Nr \log\SNR  \right\},
$$
where, 
$$
I_1(\tilde{\vect{x}},\tilde{\vect{Y}}|\mat{H}) = \log\det\big(\id_N + \frac{\SNR}{n_t}\mats{\mathcal{H}}\mats{\mathcal{H}}^\dag \big).
$$
\par
Each block $\tau_i(\tilde{\mats{\Xi}})$ contains $TNn_t$ symbols carved from a signal constellation $\mathcal{A}_d(\SNR)$. In order to maintain a rate of $R(\SNR)$ over each sub-channel, the size of the constellation $|\mathcal{A}_d(\SNR)|$ should be chosen such that, 
$$
R(\SNR) = r \log \SNR  = \frac{1}{NT} \log |\mathcal{A}_d(\SNR)|^{n_t T N}. 
$$
\textit{i.e.}, $|\mathcal{A}_d(\SNR)| = \SNR^{\frac{r}{n_t}}$. 	
It can easily be verified that for this choice of signal constellation size, the NVD parallel criterion in (\ref{eq:gen-nvd}) is, 
$$
\min_{\mat{X},\hat{\mat{X}} \in \mathcal{X}_p(\SNR)} \; \prod_{i=1}^{m} \lambda_{i}(\mat{D}\mat{D}^\dag)\dot{\geq}~\frac{1}{2^{NR(\SNR)+o(\SNR)}}.
$$
Obviously, the sufficient condition in Theorem~\ref{theo:suf-cond} is
not satisfied in this case. The achievable DMT by this transmission
scheme is only $\rho(n_t-r)(n_r-r)$ as shown in~\cite{francis2}, and
it is therefore sub-optimal.

\subsubsection{Split NVD parallel code}
\label{sec:split-NVD}
The second way we propose to split the data that guarantees to transmit a rate of
$R(\SNR)$ using a total power of $\SNR$ over each
sub-channel is shown in Figure
\ref{fig:split}. 
In this case, the total rate is split equally among
all the $NT$ slots.
Each block $\mats{\Xi}_i$ transmits $TNn_t$ symbols carved from a
signal constellation $\mathcal{A}_s(\SNR)$. The same $TNn_t$ symbols
are transmitted over blocks $\mats{\Xi}_i \; \ldots \;
\tau_{N-1}(\mats{\Xi}_{i})$ but encoded differently. However,
different symbols are transmitted over two different blocks
$\mats{\Xi}_i$ and $\mats{\Xi}_j$.
\begin{figure}[htbp]
\centering
\includegraphics[width =  \linewidth]{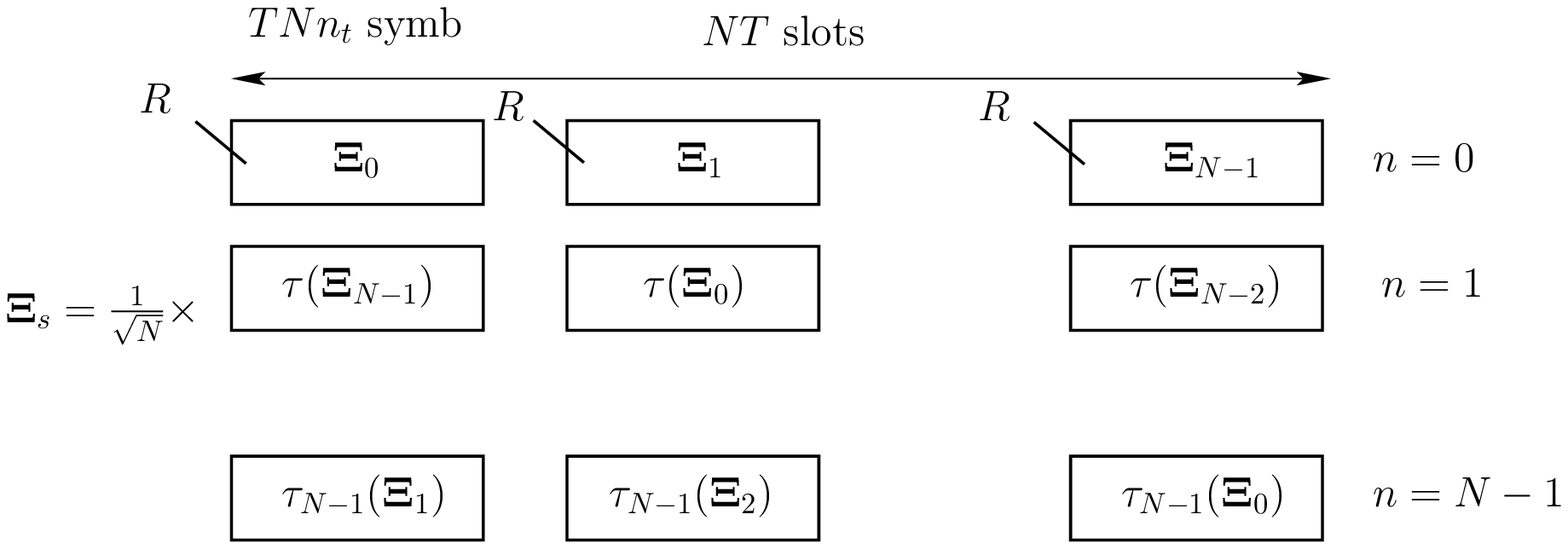}
\caption{Coding across time and frequency: The total rate is split across the $NT$ slots. Each entry of $\tau_i(\mats{\Xi}_i)$ is a linear combination of symbols carved from $\mathcal{A}_s(\SNR)$ where $|\mathcal{A}_s(\SNR)| = \SNR^{\frac{r}{Nn_t}}$. In this case, $\mat{X}_{e,s} = \theta_s \mats{\Xi}_s$.}
\label{fig:split}
\end{figure}
\par
For this transmission scheme,  the outage event occurs when at least
one of the NVD parallel code scheme with rate $R(\SNR) = r\log\SNR$ is in outage, meaning that, 
$$
\mathcal{O}_2(r,\SNR) = \bigcup_{s = 0}^{N-1} \mathcal{O}_s(r,\SNR),
$$
where, 
$$
\mathcal{O}_s (r,\SNR) = \left\{ \frac{1}{N}
I_2(\tilde{\vect{x}},\tilde{\vect{Y}}|\mat{H})  < r \log\SNR
\right\}, \quad \forall s,
$$
and,
$$
I_2(\tilde{\vect{x}},\tilde{\vect{Y}}|\mat{H}) = \log\det\big(\id_N +
\frac{\SNR}{Nn_t}\mats{\mathcal{H}}\mats{\mathcal{H}}^\dag \big).
$$
Note that the normalization factor $1/N$ in the first side of the
inequality in the outage event $\mathcal{O}_s(r,\SNR)$ traduces the fact that $N$ blocks are needed to decode the
information of each NVD parallel code with rate $R(\SNR)$. 
\\* 
Using the union bound and the inclusion bound ($\mathcal{O}_s \subseteq \mathcal{O}_2$), the outage probability can be bounded as,
\begin{equation}
\P(\mathcal{O}_s)  \leq \P(\mathcal{O}_2) \leq \sum_{i = 0}^{N-1}
\P(\mathcal{O}_s) 
\label{eq:bound-out}
\end{equation}
Assuming that $\P(\mathcal{O}_s)$ scales as $\SNR^{-d_s(r)}$, it
follows from~(\ref{eq:bound-out}) that at high SNR,
$$
 \P(\mathcal{O}_2)\doteq \SNR^{-d_s(r)} \doteq \P(\mathcal{O}_s) \doteq \P(\mathcal{O}_1),
$$
This implies that this scheme is equivalent in term
of outage to the first scheme.
\par
In order to maintain the rate of $R(\SNR)$ over each sub-channel, the signal constellation $\mathcal{A}_s(\SNR)$ should be chosen such that,
$$
R(\SNR) = r \log \SNR = \frac{1}{T} \log |\mathcal{A}_s(\SNR)|^{n_t TN}. 
$$
The size of the signal constellation for the split NVD parallel scheme is therefore reduced compared to the block diagonal case, and 
$$
|\mathcal{A}_s(\SNR)| = \SNR^{\frac{r}{Nn_t}} =
|\mathcal{A}_d(\SNR)|^{\frac{1}{N}}.
$$
Due to the block diagonal channel matrix structure, it can be deduced
that the split NVD parallel code is equivalent to a concatenation of
$N$ independent parallel NVD codes, where the symbols of each NVD
parallel code are carved from a constellation $\mathcal{A}_s(\SNR)$
with size $\SNR^{\frac{r}{Nn_t}}$.
The system is in error if at least one of the NVD parallel codes is in
error, \textit{i.e.},
$$
\varepsilon(r,\SNR) = \bigcup_{i = 0}^{N-1} \varepsilon_i(r,\SNR),
$$
where $\varepsilon(r,\SNR)$ represents the event that the system is in
error and $\varepsilon_i(r,\SNR)$ denotes the event that the $\text{i}^{\text{th}}$ NVD
parallel code formed by the blocks $\mats{\Xi}_i \; \ldots \;
\tau_{N-1}(\mats{\Xi}_{i})$ is in error. 
For each NVD parallel code with symbols carved from  $\mathcal{A}_s(\SNR)$, it can be easily verified by replacing the cardinality of $\mathcal{A}_s(\SNR)$ in~(\ref{eq:gen-nvd}) that the NVD parallel criterion in Theorem~\ref{theo:suf-cond} is satisfied, \textit{i.e.}, 
$$
\min_{\mat{X},\hat{\mat{X}} \in \mathcal{X}_p(\SNR)} \; \prod_{i=1}^{m} \lambda_{i}(\mat{D}\mat{D}^\dag)~ \dot{\geq}\frac{1}{2^{R(\SNR)+o(\SNR)}}.
$$
It follows from Theorem~\ref{theo:suf-cond} that, 
$$
\P(\varepsilon_i) \doteq \SNR^{-d_i(r)}, 
$$
where $d_i(r) = (\rho M -r)(m-r)$, $\forall i$. \\* 
Using the inclusion and the union bound as for the outage analysis in
(\ref{eq:bound-out}), it follows that, 
$$
\P_e(r,\SNR) = \P(\varepsilon) \doteq \SNR^{-d(r)}, 
$$
with $d(r) = d_i(r) = (\rho M -r)(m-r)$.
\\*
The split NVD parallel codes in Figure \ref{fig:split} achieve
therefore the optimal DMT of $(\rho M-r)(m-r)$.

\section{Numerical results} 
In order to compare the performance of the split NVD parallel code with the classical NVD parallel code, we consider the case of $2$ parallel $2 \times 2$ MIMO channel, \textit{i.e.} a block fading channel with a total number of blocks equal to $2$. 
\newline
The structure of the NVD parallel code for this configuration is given in \cite{Sheng-Par-Chan}, such that
\begin{equation}
\mat{X} =\left(\begin{array}{cc}
\mats{\Xi} & 0 \\
0 &  \tau(\mats{\Xi})
\end{array}\right)
\end{equation}
where $\mats{\Xi}$ is given in (\ref{eq:xi}) with $\theta=\frac{1+\sqrt{5}}{2}$, $\bar{\theta}=\frac{1-\sqrt{5}}{2}$,
$\alpha = 1 +i - i\theta$, $\bar{\alpha} = 1 +i - i\bar{\theta}$ and $\zeta_8=e^{\frac{i \pi}{4}}$.
The channel matrix $\tau(\mats{\Xi}$ can be deduced from $\mats{\Xi}$
by replacing $\zeta_8$ by $-\zeta_8$.

\begin{figure*}[ht]
\begin{equation}
\mats{\Xi} =\frac{1}{\sqrt{5}}\left(\begin{array}{cc}
\alpha (s_1+s_2\zeta_8+s_3\theta+s_4\zeta_8\theta) & \alpha(s_5+s_6\zeta_8+s_7\theta+s_8\zeta_8\theta)\\
\zeta_8 \bar{\alpha}(s_5+s_6\zeta_8+s_7\bar{\theta}+s_8\zeta_8\bar{\theta}) & \bar{\alpha} (s_1+s_2\zeta_8+s_3\bar{\theta}+s_4\zeta_8\bar{\theta})\end{array}\right).
\label{eq:xi}
\end{equation}
\end{figure*}

\par
For the same channel model, the structure of the split NVD parallel code is such that, 
\begin{equation}
\mat{X} = \frac{1}{\sqrt{2}}\left(\begin{array}{cc}
\mats{\Xi}_1 & \mats{\Xi}_2  \\
 \tau(\mats{\Xi}_2) &  \tau(\mats{\Xi}_1)
\end{array}\right)
\end{equation}

As we showed in previous section, the optimal DMT achievable by the NVD
parallel is only $2(2-r)(2-r)$. However, the optimal DMT achievable by the split
code is $(4-r)(2-r)$.  These two DMT are depicted in Figure \ref{fig:example-DMT}.
\begin{figure}[ht]
\centering
\includegraphics[width =
0.75\linewidth]{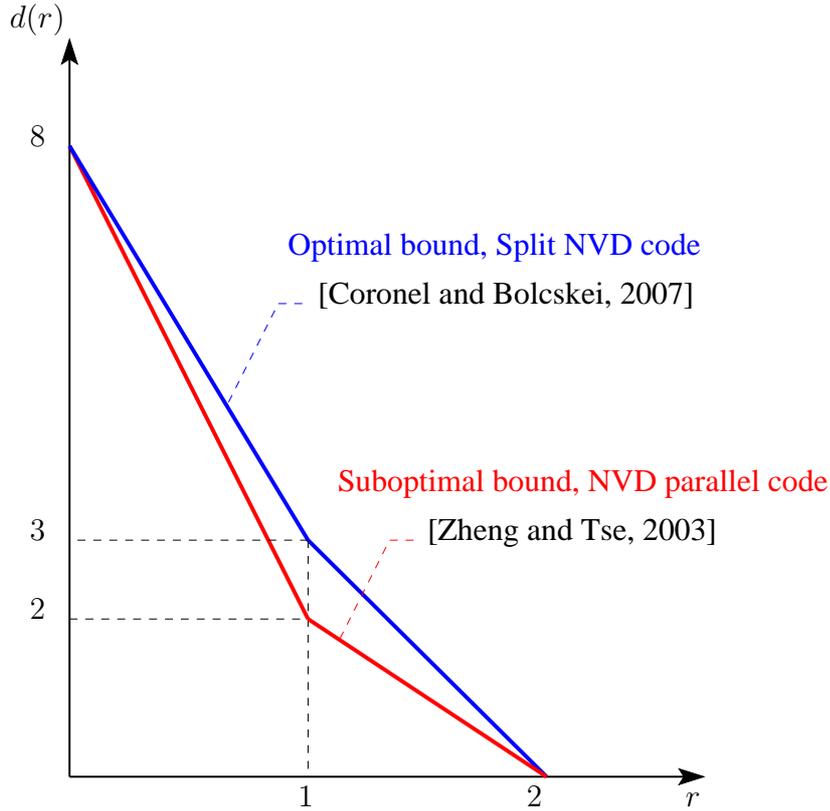}
\caption{The optimal DMT achievable by the NVD parallel code for the
  $2\times2$ block fading channel with $N = 2$ is $d(r) =
  2(2-r)(2-r)$. The split code achieves the optimal DMT of the block
  fading channel $d(r) = (4-r)(2-r)$.}
\label{fig:example-DMT}
\end{figure}

For a rate per channel use equal to $4$ bpcu (resp. $8$ bpcu), the
symbols $s_1, s_2, \ldots, s_8$ should be carved from a BPSK (resp. QPSK) constellation for the scheme with split
code and from a QPSK (resp. 16QAM) constellation for the scheme with NVD parallel
code. One should expect here that the gain provided
by the use of a smaller size of constellation used in the split NVD parallel code 
to be compensated by the normalization factor $1/\sqrt{2}$. Due to the
gain in DMT, this is not the case and the comparison of both schemes is in Figure
\ref{fig:Split_code}. 
\begin{figure}[ht]
\centering
\includegraphics[width =  0.75\linewidth,angle = 270]{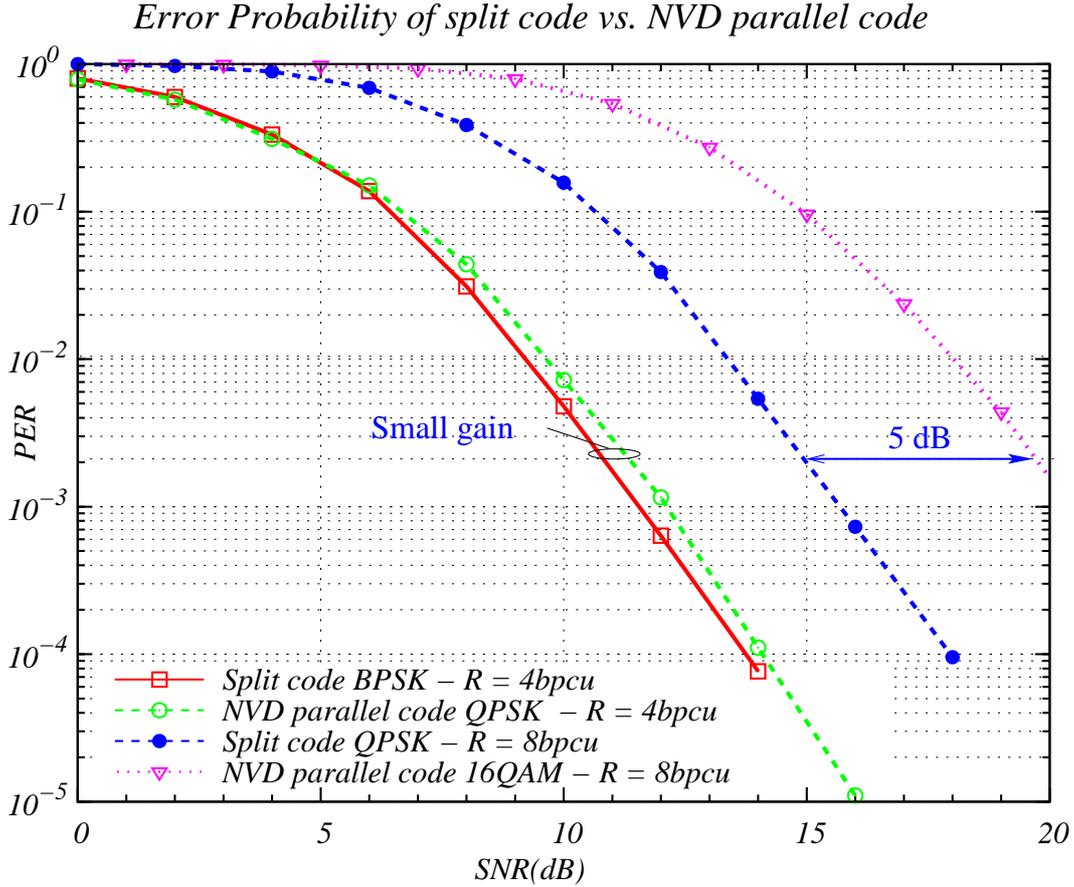}
\caption{Comparison of split NVD code versus NVD parallel code for a
  block fading MIMO channel with $N=2$ blocks and $n_t = n_r = 2$.}
\label{fig:Split_code}
\end{figure}
It can  be easily shown there the gain of the split codes
compared to the NVD parallel case is significant when the spectral
efficiency of the code increases. For a small rate of $4$ bpcu, a small gain can
be observed. However, for the rate of $8$ bpcu, approximately $5$ dB
of gain can be observed.

\section{Conclusion}
\label{sec:conclusion}
In this paper, we considered the class of selective fading MIMO
channel where the channel is selective either in time and in
frequency. Motivated by the open literature debate on the optimal
achievable DMT for the block fading channel and using completely
different arguments than~\cite{Coronel-FS} and \cite{Cor-FS-Jour},  we
proved here that the optimal DMT expression in~\cite{Coronel-FS} is
achievable for all the classes of selective fading channels, including
the block fading channel. Using the geometrical argument, we showed
that the outage bound in \cite{ZhengTse-2003} is not limiting for the outage
probability as claimed in \cite{francis1}. Moreover, a new family of split NVD
parallel codes to achieve the optimal DMT in \cite{Coronel-FS} for the
case of time or frequency selective channels is proposed. 


\appendices

\section{Proof of Theorem \ref{theo:suf-cond}}
\label{ap:proof-theo2}
Let $\mat{X}$ be the transmitted codeword, $\hat{\mat{X}}$ the nearest decoded codeword and $\Delta \mat{X}_n = \mat{X}_n - \hat{\mat{X}}_n$ the difference codeword matrix.  
The pairwise error probability of the correlated parallel channels is upper-bounded as following, 
\begin{eqnarray}
\nonumber
\PEP &\leq& \E_{\mat{H}} \exp\left(-\frac{\SNR}{4n_t} \sum_{n=0}^{N-1}\|\mat{H}_n\Delta \mat{X}_n\|^2_{\F}\right),\\
\label{eq:PEP1}
&\leq& \E_{\mat{H}} \exp\Big(-\frac{\SNR}{4n_t}\Tr\big(\mat{H}_w \mats{\Theta}\mat{H}_w^\dag \big)\Big), 
\end{eqnarray}
where $\mat{H}_w$ denotes the $n_r \times N n_t$ i.i.d. $\mathcal{CN}(0,1)$ matrix, and
$$
\mats{\Theta} = (\mat{R}_{\mathbb{H}}^{1/2} \otimes \id_{n_t})\diag\Big\{\Delta \mat{X}_n\Delta \mat{X}_n^\dag\Big\}_{n=0}^{N-1}(\mat{R}_{\mathbb{H}}^{1/2} \otimes \id_{n_t})
$$
is the effective codeword matrix. 
\\*
Assuming that $\mathcal{X}_p(\SNR)$ satisfies the NVD criteria, then $\mat{D} = \diag \big\{ \Delta \mat{X}_n\big\}_{n=0}^{N-1}$ is a full rank matrix with rank equals to $Nn_t$.  The rank and the eigenvalues of the effective codeword matrix $\mats{\Theta}$ can be computed using the following lemma~\ref{lemma:eig-lower}.
\\*
\begin{lemma}
Let $\mat{A}$ be a $p \times p$ Hermitian matrix given by, 
$$
\mat{A} = \mat{B}(\mat{CC}^\dag)\mat{B}^\dag,
$$
where $\mat{B}$ is $p\times p$ matrix with rank $s$, $\mat{C}$ is full rank $p \times p$ matrix. Then, the matrix $\mat{A}$ has the following properties:
\begin{enumerate}
\item[a)] The rank of $\mat{A}$ is equal to $s$, the rank of $\mat{B}$.
\item[b)] 
The non zero eigenvalues $\lambda_k(A)$ of $\mat{A}$ are lower bounded by, 
\begin{equation}
\lambda_k(\mat{A}) \geq \lambda_1(\mat{BB}^\dag)\lambda_k(\mat{CC}^\dag). 
\label{eq:eig-bound}
\end{equation}
\end{enumerate}
\label{lemma:eig-lower}
\end{lemma} 
 
\begin{proof}
The proof of this lemma uses the same matricial tools as \cite{Cor-FS-Jour}, and is detailed in Appendix~\ref{chap2:ap1}.
\end{proof}

By applying Lemma \ref{lemma:eig-lower}-a to $\mats{\Theta}$, it follows that,
\begin{eqnarray*}
\rank\{\mats{\Theta}\} &=& \rank\{\mat{R}_{\mathbb{H}}^{1/2} \otimes \id_{n_t}\} \\
&=& \rank\{\mat{R}_{\mathbb{H}}^{1/2}\} \rank\{\id_{n_t}\} = \rho n_t.
\end{eqnarray*}
~
By noticing that $\mats{\Theta}$ is not full rank, the Frobenius norm in (\ref{eq:PEP1}) has the same distribution as $\Tr\{\bar{\mat{H}}_w \bar{\mats{\Lambda}}\bar{\mat{H}}_w^\dag\}$ where  $\bar{\mat{H}}_w$ is the $n_r \times \rho n_t$ effective channel with i.i.d. entries $\sim \mathcal{CN}(0,1)$ and $\bar{\mats{\Lambda}}$ is the $\rho n_t \times \rho n_t$ diagonal matrix containing the non-zero eigenvalues of the effective codeword $\mats{\Theta}$ bounded using Lemma \ref{lemma:eig-lower}-b such that
$$
\lambda_i(\mats{\Theta}) \geq \sigma^2_{\mathbb{H}}\;\lambda_i\big(\mat{D}\mat{D}^\dag\big), \quad i = 1 \ldots \rho n_t,
$$
where $\sigma_{\mathbb{H}}^2$ is the smallest eigenvalue of $\mat{R}_{\mathbb{H}}$.

By following the same footsteps as in [(105) and (108) in \cite{Cor-FS-Jour}], this Frobenius norm can be bounded such that, 
\begin{eqnarray*}
\Tr\{\bar{\mat{H}}_w \bar{\mats{\Lambda}}\bar{\mat{H}}_w^\dag\} &\geq& \sum_{i=1}^{m} \lambda_i(\Jensen{H}_w\Jensen{H}_w^\dag)\lambda_{m-i+1}(\mats{\Theta}) \\
&\geq& \sigma_{\mathbb{H}}^2 \sum_{i=1}^{m}\lambda_i(\Jensen{H}_w\Jensen{H}_w^\dag)\lambda_{m-i+1}(\mat{DD}^\dag)
\end{eqnarray*}
where $\Jensen{H}_w$ denotes the $m \times \rho M$ Jensen channel with i.i.d. $\mathcal{CN}(0,1)$ entries such that,
\begin{equation}
\Jensen{H}_w = 
\begin{cases}
[\mat{H}_{w,0} \; \ldots \; \mat{H}_{w,\rho -1}], \qquad \text{if}~n_r \leq n_t, \\
[\mat{H}_{w,0}^{\dag} \; \ldots \; \mat{H}_{w,\rho -1}^\dag], \qquad \text{if}~n_r > n_t.
\end{cases}
\label{eq:Jensen}
\end{equation}

The rest of the proof uses the same technique as presented in \cite{Coronel-FS},\cite{Cor-FS-Jour}. It can be deduced that if the code satisfies the NVD criteria in (\ref{eq:nvd-o}), then the  error region event $E_{\vects{\vects{\alpha}}}(r,\SNR)$ for a given channel realisation $\vects{\alpha}$ matches with the outage region $\mathcal{O}_{\vects{\alpha}}^{[m,\rho M]}(r,\SNR)$ of the equivalent $m \times \rho M$ MIMO channel,
\begin{align}
\label{eq:error-reg}
E_{\vects{\alpha}}(r,\SNR) &= \Big\{ \sum_{i=1}^{k} \alpha_i \geq k-r, \;\; k = 1, \ldots, m \Big\},\notag \\ 
&= \mathcal{O}_{\vects{\alpha}}^{[m,\rho M]}(r,\SNR),
\end{align}
with $\vects{\alpha}$ being the vector containing the eigen exponents of the channel $\Jensen{H}_w\Jensen{H}_w^\dag$, such that $\lambda_i(\Jensen{H}_w\Jensen{H}_w^\dag) \doteq \SNR^{-\alpha_i}$.

\section{Proof of Lemma \ref{lemma:eig-lower}}
\label{chap2:ap1}
As $\mat{A}$ is an Hermitian matrix, its rank is equal to the rank of $\mat{BC}$. It can be easily checked from the product matrix rank property in (\ref{eq:prod-rank}), ($\mat{D}\in \mathbb{C}^{a\times b}, \mat{E}\in \mathbb{C}^{b \times c}$), 
\begin{align}
\rank\{\mat{D}\} + \rank\{\mat{E}\}&  - b \leq \rank\{\mat{D}\mat{E}\} \notag \\ & \leq \min\big\{\rank\{\mat{D}\},\rank\{\mat{E}\}\big\},
\label{eq:prod-rank}
\end{align}
and the fact that $\mat{C}$ is a full rank matrix, that, 
$$
\rank\{\mat{B}\} + p - p\leq \rank\{\mat{A}\} \leq \rank\{\mat{B}\}, 
$$
which implies that, 
$
\rank\{\mat{A}\} = \rank\{\mat{B}\}.
$
\par
Using the fact that for a square matrix $\mat{M} \in \mathbb{C}^{a \times a}$,
$\lambda(\mat{MM}^\dag)= \lambda(\mat{M}^\dag\mat{M}),$
implies that
$$
\lambda_k(\mat{A})  = \lambda_k(\mat{C}^\dag\mat{B}^\dag\mat{BC}).
$$
Let $\mat{B}^\dag\mat{B} = \mat{U}\mats{\Lambda}\mat{U}^\dag$ be the eigenvalue decomposition of $\mat{B}^\dag\mat{B}$, with $\mats{\Lambda} = [\mats{\tilde{\Lambda}}\; \mat{0}_{p-s}]$. Then, 
\begin{subequations}
\begin{eqnarray*}
\label{eq:A1}
\lambda_k(\mat{A})  &=& \lambda_k(\mat{C}^\dag\mat{U\Lambda U}^\dag\mat{C}), \\
\label{eq:A2}
 &=& \lambda_k(\mats{\Lambda}^{1/2}\mat{U}^\dag\mat{CC}^\dag \mat{U} \mats{\Lambda}^{1/2}).
\end{eqnarray*}
Let $\mats{\Omega} = \mat{U}^\dag(\mat{CC}^\dag) \mat{U}$ and $\mats{\tilde{\Omega}}$ be the $s \times s$ principal submatrix of $\mats{\Omega}$.  Then, 
\begin{eqnarray}
\label{eq:A3}
\lambda_k(\mat{A}) &=& \lambda_k(\mats{\Lambda}^{1/2}\mats{\Omega}\mats{\Lambda}^{1/2}), \\
\label{eq:A4}
 &=& \lambda_k(\mats{\tilde{\Lambda}}^{1/2}\mats{\tilde{\Omega}}\mats{\tilde{\Lambda}}^{1/2}),
\end{eqnarray}
As $\mats{\tilde{\Lambda}}^{1/2}$ in (\ref{eq:A4}) is non singular matrix and $\mats{\tilde{\Omega}}$ is Hermitian, The Ostrowski theorem in \cite{Johnson} can be applied, 
\begin{eqnarray}
\label{eq:A5}
\lambda_k(\mat{A}) &\geq&  \lambda_1(\mats{\tilde{\Lambda}}) \lambda_k(\mats{\tilde{\Omega}}), \\
\label{eq:A6}
&\geq& \lambda_{1}(\mat{BB}^\dag) \lambda_k(\mats{\Omega}), \\
\label{eq:A7}
&=&  \lambda_1(\mat{BB}^\dag) \lambda_k(\mat{CC}^\dag).
\end{eqnarray}
As $\mats{\tilde{\Omega}}$ is a $s \times s$ submatrix of the Hermitian matrix $\mats{\Omega}$, (\ref{eq:A6}) follows from the application of theorem 4.3.15 in \cite{Johnson}. Finally, (\ref{eq:A7}) follows from the fact that $\mat{U}$ is unitary matrix, and therefore 
$
\lambda_k(\mats{\Omega}) =\lambda_k(\mat{CC}^\dag). 
$
\end{subequations}

\end{document}